\documentclass[journal]{IEEEtran}

\usepackage{amsmath,amssymb,amsfonts}
\usepackage{bm}
\usepackage{graphicx}
\usepackage{booktabs}
\usepackage{cite}
\usepackage{algorithm}
\usepackage{algpseudocode}
\usepackage{array}
\usepackage{url}
\usepackage{balance}

\AtBeginDocument{%
  \setlength{\abovedisplayskip}{4pt plus 2pt minus 2pt}%
  \setlength{\belowdisplayskip}{4pt plus 2pt minus 2pt}%
  \setlength{\abovedisplayshortskip}{2pt plus 1pt minus 1pt}%
  \setlength{\belowdisplayshortskip}{2pt plus 1pt minus 1pt}%
}

\newtheorem{proposition}{Proposition}

\newcommand{\ind}{\mathbf{1}}

\allowdisplaybreaks

\begin{document}
\raggedbottom

\title{Q-MAP: Multi-Platform Benchmarking of Distributed Quantum Computing for Coherent Controlled Islanding}

\author{Yuqi~Jiang,~\IEEEmembership{Graduate Student Member,~IEEE},
        Zhiding~Liang,~\IEEEmembership{Member,~IEEE},
        Qiang~Guan,~\IEEEmembership{Senior Member,~IEEE},
        Yan~Li,~\IEEEmembership{Senior Member,~IEEE},
        and~Ganesh~Kumar~Venayagamoorthy,~\IEEEmembership{Fellow,~IEEE}%
\thanks{This work is supported by the Office of Naval Research under award
N00014-22-1-2504 and the National Science Foundation under awards OAC-2417773
and ECCS-2413237/2413238.}%
\thanks{Y. Jiang and Y. Li are with the Department of Electrical Engineering,
The Pennsylvania State University, University Park, PA 16802, USA.}%
\thanks{Z. Liang is with the Department of Computer Science,
Rensselaer Polytechnic Institute, Troy, NY 12180, USA.}%
\thanks{Q. Guan is with the Department of Computer Science,
Kent State University, Kent, OH 44242, USA.}%
\thanks{G. K. Venayagamoorthy is with the Real-Time Power and Intelligent
Systems Laboratory, Holcombe Department of Electrical and Computer Engineering,
Clemson University, Clemson, SC 29634, USA, and with the University of Pretoria,
Pretoria, South Africa.}%
\thanks{Corresponding author: Yan Li (e-mail: yql5925@psu.edu).}}

\maketitle

\begin{abstract}
The integration of distributed energy resources into power networks is
accelerating. The resulting variability narrows operating margins, so a
disturbance can cascade into a wide-area blackout.
Controlled islanding arrests that propagation by splitting a compromised grid
into self-sustaining islands that keep coherent generators together. Exact
classical solutions become intractable as the bus count and island number
grow. Gate-based quantum optimization provides a different route through this
combinatorial space, although its reach is limited when one circuit carries
every bus assignment, since qubit count and depth then follow grid size. In
this study, a round-synchronous distributed quantum computing framework is
developed for coherent controlled islanding under a fixed per-circuit qubit
budget. Every round derives all regional subproblems from one frozen grid-wide
snapshot, dispatches them to independent quantum backends at the same time,
and merges the returned candidates classically into one globally evaluated
update. Circuits executed
in parallel therefore keep a constant size as the grid grows, and a round
costs the slowest region rather than the sum of all of them. Benchmarking
spans IEEE systems from 9 to 300 buses on simulation and on quantum processors
of different architectures. The framework attains optimal and operationally
feasible partitions on every platform under noise, even where compilation cost
differs by nearly an order of magnitude. Bounding width in this way places grids
beyond the reach of monolithic circuits within range of present devices and
establishes a multi-backend baseline for quantum computing in large-scale
power-system optimization.
\end{abstract}


\section{Introduction}
\label{sec:introduction}

The growing integration of distributed energy resources introduces new
operational challenges, including the prevention of cascading
failures~\cite{yi2021aggregate,vaiman2012risk}. Their uncertainty, heterogeneity,
and dispersed coordination complicate secure operation under stressed
conditions~\cite{yi2021aggregate}. Following an initial outage, protection
actions and redistributed power flows can trigger successive disconnections and
propagate disturbances across the network~\cite{vaiman2012risk}. Rapid
containment is therefore essential to prevent widespread blackouts.

Controlled islanding contains severe disturbances by opening selected lines
before cascading failures spread~\cite{ahmed2003scheme,kyriacou2018controlled}.
The resulting islands must preserve coherent generators and sufficient
operating resources while minimizing disruption. Identifying such a partition
quickly is therefore a constrained combinatorial optimization problem.

Classical solutions use slow coherency, graph-based search, and mathematical
programming. Slow-coherency methods identify compatible generator groups but
require a separate cutset search to assign the remaining
buses~\cite{you2004slow,yang2006slowcoherency}. Ordered binary decision diagram
methods search feasible splits systematically but scale
combinatorially~\cite{sun2003splitting}, whereas
spectral clustering gains efficiency through relaxation while often enforcing
only limited feasibility conditions~\cite{ding2013twostep}. Mixed-integer models
represent operational constraints more explicitly~\cite{teymouri2019toward,patsakis2019strong},
but their computational burden grows with the binary model. Classical methods
therefore scale poorly, making global optimality increasingly difficult to
guarantee as system size grows.

Quantum optimization offers a potential approach to NP-hard combinatorial
problems by encoding an objective as a Hamiltonian whose low-energy states
represent favorable solutions. The quantum approximate optimization algorithm
(QAOA) prepares and variationally optimizes such
states through alternating problem and mixing evolutions~\cite{zhou2020qaoa}.
Implementations on superconducting processors have demonstrated variational
quantum optimization on nonplanar graph problems, supporting the potential of
QAOA for challenging combinatorial tasks~\cite{harrigan2021qaoa}.

Quantum computing has also reached diverse application domains. Quantum
algorithms have been developed for financial risk
analysis~\cite{woerner2019risk}, molecular ground-state estimation in
chemistry~\cite{peruzzo2014variational}, protein folding in computational
biology~\cite{robert2021protein}, and quantum-enhanced machine
learning~\cite{havlicek2019supervised}. This breadth provides a foundation for
extending quantum methods to complex infrastructure optimization.

Within critical infrastructure, quantum methods have recently been applied to
power-system monitoring, operation, and network partitioning. Representative
studies address optimal phasor measurement unit
placement~\cite{jiang2025pmu}, stochastic economic
dispatch~\cite{han2025dispatch}, and power-grid
partitioning~\cite{hartmann2025quantum}. The partitioning approach relies on
quantum annealing rather than gate-based QAOA. It therefore leaves the QAOA
circuit design and the operational constraints required for controlled
islanding unresolved. Existing QAOA frameworks incorporate islanding
constraints through specialized encoding and postprocessing
strategies~\cite{jiang2026regrid,jiang2026pace}. Their circuit width, however,
still grows with the number of optimized buses. Since usable qubits remain
scarce on current noisy intermediate-scale quantum devices, this scaling
limits their applicability to large networks and motivates a distributed
formulation with bounded circuit width.

In this study, a
round-synchronous distributed QAOA framework, is proposed to
solve coherent controlled islanding under a prescribed per-circuit qubit
budget. The global problem is decomposed into qubit-bounded quadratic
unconstrained binary optimization (QUBO) subproblems that are solved
concurrently from a common frozen assignment. A classical
coordinator combines the regional candidate pools through feasibility-aware
beam search, constrained postprocessing, and monotonic global updates. This
architecture bounds the width of each quantum circuit independently of network
size while retaining global coordination across the complete islanding
problem. In this study, the framework is tested on the IEEE 9- to 300-bus test
systems across six quantum-computing environments, ranging from ideal and
calibrated-noise simulation to the IBM, IQM, and Rigetti processors. The
numerical examples verify the effectiveness of the proposed method and its
noise resilience across these backends. The main contributions of this paper
can be summarized as follows.
\begin{itemize}
\item A round-synchronous distributed QAOA framework is proposed that keeps
the quantum effort local and width-bounded while the coordination remains
global and classical, so the circuit size no longer scales with the network.
\item The framework recovers the optimal islanding solution on the IEEE 9- to
300-bus systems with at most 20 qubits per circuit, and that quality is
retained under calibrated noise and on real quantum processors, which
demonstrates the noise resilience of the proposed method.
\item The evaluation across six execution environments establishes a
pioneering multi-backend baseline for gate-based quantum optimization in
power-system applications, quantifying the depth and gate cost that each
platform imposes on the same workload.
\end{itemize}

\section{Constrained Bus-Partitioning Formulation}
\label{sec:formulation}

Controlled islanding partitions a power network into electrically separated,
self-sustaining islands while minimizing disruption. This task is formulated
as a constrained bus-partitioning problem. The constraints ensure that each
resulting island forms a feasible subsystem.

\subsection{Network Model and Disruption Objective}
\label{subsec:network}

Let $\mathcal P=(\mathcal B,\mathcal L)$ denote the pre-islanding power
network, modeled as a connected undirected graph with bus set $\mathcal B$,
$n=|\mathcal B|$, and line set $\mathcal L$. The goal is to assign every bus
to one of $K\ge2$ islands indexed by $\Omega=\{1,\ldots,K\}$; none of the
formulation below is restricted to the binary case. Following the
disruption-weight convention in~\cite{jiang2026pace}, each line
$\{i,j\}\in\mathcal L$ is weighted by the average magnitude of its
pre-islanding directional active-power flows $P^{0}_{ij}$ and $P^{0}_{ji}$,
\begin{equation}
\rho_{ij}=\tfrac12\big(|P^{0}_{ij}|+|P^{0}_{ji}|\big),
\qquad \{i,j\}\in\mathcal L,
\label{eq:line-weight}
\end{equation}
so that opening a heavily loaded interface incurs a larger penalty than
opening a lightly loaded one. All active-power quantities in this paper are
expressed in MW, so the line weights $\rho_{ij}$ and every objective value
reported in Section~\ref{sec:results} carry that unit.

A candidate partition is a label $\ell_i\in\Omega$ for every bus, collected
into $\bm\ell=(\ell_1,\ldots,\ell_n)\in\Omega^{n}$, with membership indicator
$u_{i,k}=\ind\{\ell_i=k\}$, where $\ind\{\cdot\}$ equals one when its argument
holds and zero otherwise. The disruption objective totals the weight of
every line whose two endpoints disagree,
\begin{equation}
\Phi(\bm\ell)=\!\!\sum_{\{i,j\}\in\mathcal L}\!\!\rho_{ij}\ind\{\ell_i\neq\ell_j\}
=\!\!\sum_{\{i,j\}\in\mathcal L}\!\!\rho_{ij}
\Big(1-\sum_{k\in\Omega}u_{i,k}u_{j,k}\Big).
\label{eq:objective}
\end{equation}
The objective alone would favor partitions that avoid heavily loaded cuts,
but it does not ensure that the resulting subsystems are usable islands.
Generator and demand buses are denoted by $\mathcal B_{\mathrm g}$ and
$\mathcal B_{\mathrm d}$, with both sets contained in $\mathcal B$; the two
sets need not be disjoint. In addition, $K$ nonempty, pairwise-disjoint
coherent generator groups $\Gamma_1,\ldots,\Gamma_K\subseteq\mathcal
B_{\mathrm g}$, indexed by $j\in\{1,\ldots,K\}$, are given. A coherent group collects generator buses whose
swing dynamics require them to remain in a common island, and the following
feasibility requirements turn the cut-minimization objective into a
controlled-islanding model.

\subsection{Island Feasibility Conditions}
\label{subsec:feasibility}

An admissible label vector must satisfy four requirements. They are stated on
the complete assignment $\bm\ell$, because several of them cannot be certified
from a local bus decision or from an isolated regional subproblem.

\smallskip\noindent\textbf{Island connectivity.} For $k\in\Omega$, let
\begin{equation}
\mathcal B_k(\bm\ell)=\{i\in\mathcal B:\ell_i=k\},
\qquad
\mathcal P_k(\bm\ell)=\mathcal P[\mathcal B_k(\bm\ell)]
\label{eq:island-subgraph}
\end{equation}
be its bus set and induced subgraph, and let $c_k(\bm\ell)$ count the
connected components of $\mathcal P_k(\bm\ell)$, with $c_k(\bm\ell)=0$ when
$\mathcal B_k(\bm\ell)$ is empty.
Every resulting island must be a single contiguous piece of the network:
\begin{equation}
c_k(\bm\ell)=1,\qquad \forall k\in\Omega.
\label{eq:connectivity}
\end{equation}
Because this condition depends on the
full set of buses assigned to island $k$, it cannot be verified from any
single bus in isolation. It is checked, by depth-first search on
$\mathcal P_k(\bm\ell)$, only once a complete label vector has been formed.

\smallskip\noindent\textbf{Coherency and separation.} Connectivity alone does
not preserve generator dynamics. Every coherent group must stay together, and
distinct groups must end up in distinct islands,
\begin{subequations}
\label{eq:coherency-separation}
\begin{align}
\ell_i&=\ell_{i'}, &&\forall i,i'\in\Gamma_j,
\label{eq:coherency}\\
\ell_i&\neq\ell_{i'}, &&\forall i\in\Gamma_j,\ \forall i'\in\Gamma_{j'},\
j\neq j',
\label{eq:separation}
\end{align}
\end{subequations}
for all $j,j'\in\{1,\ldots,K\}$.
Since there are exactly $K$ groups and $K$ islands,
\eqref{eq:coherency-separation} forces the map from groups to islands to be a
bijection, without specifying which group is assigned to which island.

\smallskip\noindent\textbf{Minimum size and resource presence.} After coherent
groups are separated, each island must also contain enough buses and basic
resources to operate as a subsystem. With $\mu_{\mathrm b}$,
$\mu_{\mathrm g}$, and $\mu_{\mathrm d}$ denoting the required minimum bus,
generator, and demand counts per island,
\begin{equation}
|\mathcal B_k(\bm\ell)|\ge\mu_{\mathrm b},\quad
|\mathcal B_k(\bm\ell)\cap\mathcal B_{\mathrm g}|\ge\mu_{\mathrm g},\quad
|\mathcal B_k(\bm\ell)\cap\mathcal B_{\mathrm d}|\ge\mu_{\mathrm d},
\label{eq:presence}
\end{equation}
for every $k\in\Omega$. $\mu_{\mathrm b}$ is a configurable system parameter,
set to two buses for every benchmark system considered in this work. The
generator and demand thresholds are fixed at $\mu_{\mathrm g}=\mu_{\mathrm
d}=1$, so each island is required only to contain some generation and some
demand. The QUBO and dual formulation of Section~\ref{sec:method} extend to a
larger threshold without any structural change.

\smallskip\noindent\textbf{Full coverage.} Finally, the above requirements are
interpreted under a complete partition: every bus carries exactly one label,
\begin{equation}
\sum_{k\in\Omega}u_{i,k}=1,\qquad\forall i\in\mathcal B,
\label{eq:coverage}
\end{equation}
which holds identically once $u_{i,k}=\ind\{\ell_i=k\}$ is induced from a
label vector $\bm\ell\in\Omega^n$ as in \eqref{eq:objective}. Thus,
\eqref{eq:coverage} is not an additional constraint in the mathematical
program below. It is nevertheless recorded here because it reappears in
Section~\ref{sec:method} as a decoding condition on the quantum register that
represents $\ell_i$, where an all-zero or multiply-active register would
otherwise be ambiguous.

\subsection{Controlled-Islanding Optimization Problem}
\label{subsec:program}

The preceding objective and admissibility conditions define the exact target
problem addressed in this paper. Their combination gives
\begin{equation}
\Phi^{\star}=\min_{\bm\ell\in\Omega^{n}}\ \Phi(\bm\ell)
\quad\text{s.t.}\quad
\bm\ell\ \text{satisfies}\ \eqref{eq:connectivity},\quad
\eqref{eq:coherency-separation},\quad \text{and}\quad \eqref{eq:presence}.
\label{eq:program}
\end{equation}
Its feasible set is
\begin{equation}
\mathcal X=\{\bm\ell\in\Omega^{n}:\bm\ell\ \text{is feasible for \eqref{eq:program}}\},
\label{eq:feasible-set}
\end{equation}
assumed nonempty. Feasibility and solution quality throughout this paper are
defined with respect to $\mathcal X$ and the objective \eqref{eq:objective}
evaluated on the complete network. This convention is important for the
distributed method: the bounded subproblems introduced in
Section~\ref{sec:method} provide a tractable route to \eqref{eq:program}, but
they do not redefine $\mathcal X$ or replace the global objective $\Phi$.

\section{A Round-Synchronous Distributed QAOA Architecture}
\label{sec:method}

The proposed architecture decomposes the global islanding problem into
bounded regional QUBOs, each mapped to a compact QAOA circuit and executed
concurrently on an independent backend. A classical coordinator maintains the
global assignment, while stateless workers solve the regional subproblems in
a star topology. In each synchronous round, all workers use the same frozen
assignment, and their candidate solutions are aggregated into a single global
update. Fig.~\ref{fig:workflow} summarizes the resulting architecture, from
the preprocessing that fixes the job set once to the per-round cycle of
concurrent regional solves, coordinator-side aggregation, and dual update.
This cycle requires a qubit-efficient representation of every bus
label that remains variable.

\begin{figure*}[!t]
\centering
\includegraphics[width=0.92\textwidth]{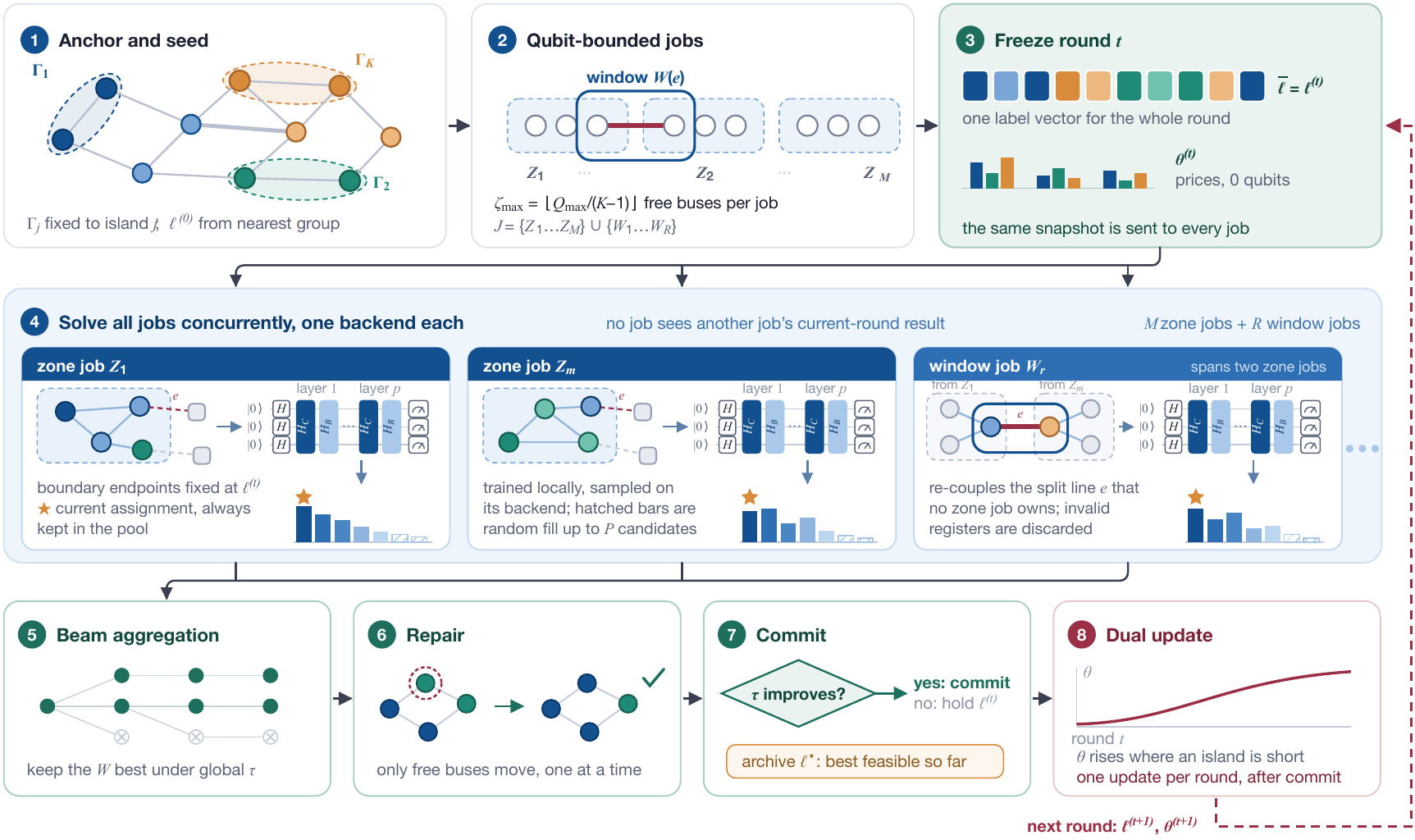}
\caption{Round-synchronous distributed QAOA workflow for coherent controlled
islanding. Preprocessing anchors the coherent groups, initializes the bus
labels, and forms qubit-bounded regional jobs (1)--(2). Each round freezes the
global state, solves all jobs concurrently, aggregates and repairs their
candidates, commits an improving assignment, and updates the dual prices
(3)--(8).}
\label{fig:workflow}
\end{figure*}

\subsection{Compact Free-Bus Encoding}
\label{subsec:encoding}

A compact encoding represents the island assignment of each free bus with
$K-1$ qubits~\cite{jiang2026pace}. For free bus $i$, the membership variables
are defined from the register
$\bm q_i=(q_{i,1},\ldots,q_{i,K-1})$ as
\begin{equation}
\tilde u_{i,k}=
\begin{cases}
q_{i,k}, & k\le K-1,\\[2pt]
1-\displaystyle\sum_{k'=1}^{K-1}q_{i,k'}, & k=K,
\end{cases}
\label{eq:decode}
\end{equation}
where the all-zero register represents label $K$, and the validity condition
$\sum_{k'=1}^{K-1}q_{i,k'}\le1$ makes the decoded label unique. Within the QUBO, invalid register states are penalized by
\begin{equation}
\Xi_i(\bm q_i)=\xi\sum_{1\le k<k'\le K-1}q_{i,k}q_{i,k'},
\qquad \xi>0,
\label{eq:reduced-penalty}
\end{equation}
which is zero for a valid register. Thus, each free bus contributes $K-1$
qubits to a regional QUBO, establishing the per-bus width used to partition
the global problem under a circuit budget.

\subsection{Coherent-Group Anchoring and Regional Decomposition}
\label{subsec:decomp}

The global islanding problem is converted into qubit-bounded regional
subproblems by anchoring coherent groups and partitioning the remaining free
buses. Because
\eqref{eq:coherency-separation} forces a bijection between the $K$
coherent groups and the $K$ islands, and island labels carry no physical
meaning of their own, fixing that bijection to the identity discards no
physically distinct partition~\cite{jiang2026pace}. With $t$ indexing the
synchronous rounds,
\begin{subequations}
\label{eq:anchors}
\begin{align}
\ell_i&=j,\qquad
\forall i\in\Gamma_j,\ \forall j\in\{1,\ldots,K\},
\label{eq:anchor-choice}\\
\mathcal B_{\mathrm a}&=\bigcup_{j=1}^{K}\Gamma_j,\qquad
\mathcal B_{\mathrm v}=\mathcal B\setminus\mathcal B_{\mathrm a},
\label{eq:anchor-sets}\\
\ell_i^{(t)}&=j,\qquad
\forall i\in\Gamma_j,\ \forall j\in\{1,\ldots,K\},\ \forall t\ge0.
\label{eq:anchor-invariant}
\end{align}
\end{subequations}
Equations~\eqref{eq:anchor-choice}--\eqref{eq:anchor-invariant} assign each
coherent group to a distinct fixed label and preserve that assignment
throughout the algorithm. Consequently, all regional jobs and repair
operations act only on $\mathcal B_{\mathrm v}$.

A complete first label vector $\bm\ell^{(0)}$ is built by a multi-source
weighted shortest-path search rooted simultaneously at every $\Gamma_j$,
using line length $w_{ij}$ and its induced distance to a group,
\begin{subequations}
\label{eq:init}
\begin{align}
w_{ij}&=1+\frac{1}{\max(\rho_{ij},\epsilon)},\qquad\epsilon=10^{-9},
\label{eq:graph-distance}\\
w(i,\Gamma_j)&=\min_{\pi:\,i\to\Gamma_j}\ \sum_{\{i',j'\}\in\pi}w_{i'j'},
\label{eq:shortest-distance}\\
\ell_i^{(0)}&=\min\operatorname*{arg\,min}_{j\in\Omega} w(i,\Gamma_j),
\qquad\forall i\in\mathcal B_{\mathrm v},
\label{eq:init-label}
\end{align}
\end{subequations}
where the minimum in \eqref{eq:shortest-distance} runs over paths $\pi$ in
$\mathcal P$ from $i$ to any bus of $\Gamma_j$, so topological proximity
and line weights determine the initial assignment. The outer minimum in
\eqref{eq:init-label} resolves exact ties by label index.

Let $n_{\mathrm v}=|\mathcal B_{\mathrm v}|$ and
$\mathcal B_{\mathrm v}=\{b_1<\cdots<b_{n_{\mathrm v}}\}$ follow bus-index
order. Under the circuit budget $Q_{\max}$, the admissible number of
variable buses per job is $\zeta_{\max}$, yielding
$M=\lceil n_{\mathrm v}/\zeta_{\max}\rceil$ disjoint zones defined by
\begin{subequations}
\label{eq:decomposition}
\begin{align}
\zeta_{\max}&=\Big\lfloor\frac{Q_{\max}}{K-1}\Big\rfloor\ge1,
\label{eq:zeta-max}\\
\mathcal Z_m&=\{b_{(m-1)\zeta_{\max}+1},\ldots,
b_{\min(m\zeta_{\max},\,n_{\mathrm v})}\},
\label{eq:zones}\\
\mathcal B_{\mathrm v}&=\biguplus_{m=1}^{M}\mathcal Z_m,\qquad
|\mathcal Z_m|(K-1)\le Q_{\max},\ \forall m,
\label{eq:zone-bound}
\end{align}
\end{subequations}
for $m=1,\ldots,M$.
Because each bus contributes $K-1$ qubits, this construction caps the regional
circuit width at $Q_{\max}$ independently of $n$, compared with
$Q_{\mathrm G}=|\mathcal B_{\mathrm v}|(K-1)
=(n-|\mathcal B_{\mathrm a}|)(K-1)$ qubits for a single global circuit.

The disjoint zones omit couplings across their boundaries. To retain the most
influential cross-zone terms under the same width bound, define
\begin{equation}
\begin{aligned}
\mathcal L_{\times}=\big\{\{i,j\}\in\mathcal L:\ &
\exists m,m'\in\{1,\ldots,M\},\\[-2pt]
&m\ne m',\ i\in\mathcal Z_m,\ j\in\mathcal Z_{m'}\big\}.
\end{aligned}
\label{eq:crossing-lines}
\end{equation}
Let $\mathcal N(i)=\{j\in\mathcal B:\{i,j\}\in\mathcal L\}$ denote the
neighbors of bus $i$. Processing the lines in $\mathcal L_{\times}$ by
decreasing $\rho$, with bus indices resolving ties, a line $e=\{i,j\}$
generates the correction window
\begin{equation}
\mathcal W(e)=\operatorname{first}_{\zeta_{\max}}\!
\Big((i,j)\,\big\Vert\,\operatorname{sort}\big((\mathcal N(i)\cup\mathcal N(j))
\cap\mathcal B_{\mathrm v}\big)\Big),
\label{eq:window}
\end{equation}
where $\Vert$ denotes ordered concatenation, $\operatorname{sort}$ uses bus
index, and $\operatorname{first}_{\zeta_{\max}}$ retains the first
$\zeta_{\max}$ distinct entries. Windows containing fewer than two buses and
duplicate windows are removed. Denote the retained windows by
$\mathcal W_1,\ldots,\mathcal W_R$. The resulting job set
\begin{equation}
\mathcal J=\{\mathcal Z_1,\ldots,\mathcal Z_M\}\cup\{\mathcal W_1,\ldots,\mathcal W_R\},
\label{eq:job-set}
\end{equation}
is fixed for all rounds. Although this decomposition bounds the regional
circuit size, the island-level resource requirements in \eqref{eq:presence}
remain global. Their influence enters each regional QUBO through a common
price vector.

\subsection{Regional QUBO Formulation}
\label{subsec:qubo}

Each round-local QUBO combines the regional cut cost, encoding penalty,
inertia, and the current dual reward. At the beginning of round $t$, the
coordinator supplies every job with the same fixed nonnegative price vector
$\bm\theta^{(t)}=(\theta_k^{h,(t)})_{h\in\mathcal H,\,k\in\Omega}$, where
$\mathcal H=\{\mathrm b,\mathrm g,\mathrm d\}$ indexes buses, generators, and
demand, respectively, and $\mathcal B_{\mathrm b}:=\mathcal B$. Assigning free
bus $i$ to island $k$ contributes the linear coefficient
\begin{equation}
\delta_{i,k}(\bm\theta)=
-\sum_{h\in\mathcal H}\ind\{i\in\mathcal B_h\}\theta_k^{h}.
\label{eq:dual-reward}
\end{equation}
This coefficient is evaluated at $\bm\theta^{(t)}$ in round $t$. It adds no qubits,
and $\bm\theta^{(t)}=\bm0$ recovers the unpriced QUBO, which is the state at
$\bm\theta^{(0)}=\bm0$ in the first round. All jobs are constructed from
the same frozen snapshot
$\bar{\bm\ell}:=\bm\ell^{(t)}$, producing a Jacobi-style update in which no job
sees another job's current-round result~\cite{jiang2026pace}. For a job
$\mathcal U\in\mathcal J$, let $\mathcal L_{\mathcal U}$ and
$\partial\mathcal U$ denote its internal and boundary lines, respectively.
With each external boundary endpoint fixed at $\bar{\bm\ell}$, the exact cut
contribution of all lines incident to $\mathcal U$ is
\begin{subequations}
\label{eq:job-cut-terms}
\begin{align}
\mathcal L_{\mathcal U}
&=\{\{i,j\}\in\mathcal L:i,j\in\mathcal U\},
\label{eq:internal-lines}\\
\partial\mathcal U
&=\{\{i,j\}\in\mathcal L:|\{i,j\}\cap\mathcal U|=1\},
\label{eq:job-lines}\\
E^{\mathrm c}_{\mathcal U}(\bm q_{\mathcal U}\mid\bar{\bm\ell})
&\nonumber\\
&={}
\sum_{\{i,j\}\in\mathcal L_{\mathcal U}}\!\!
\rho_{ij}\Big(1-\sum_{k\in\Omega}\tilde u_{i,k}\tilde u_{j,k}\Big)
\nonumber\\
&\quad{}+\!\!\sum_{\substack{\{i,j\}\in\partial\mathcal U\\ i\in\mathcal U}}\!\!
\rho_{ij}\big(1-\tilde u_{i,\bar\ell_j}\big).
\label{eq:job-cut}
\end{align}
\end{subequations}
Fixing the external labels preserves every line contribution without
introducing cross-job qubits. The encoding penalty restricted to
$\mathcal U$ is
\begin{equation}
\Xi_{\mathcal U}(\bm q_{\mathcal U})=
\sum_{i\in\mathcal U}\Xi_i(\bm q_i),
\label{eq:regional-encoding-penalty}
\end{equation}
Combining the cut cost, an inertia weight $\iota\ge0$, the encoding penalty,
and the dual reward gives
\begin{subequations}
\label{eq:job-qubo-terms}
\begin{align}
E^{\mathrm i}_{\mathcal U}(\bm q_{\mathcal U}\mid\bar{\bm\ell})&=
\iota\sum_{i\in\mathcal U}\big(1-\tilde u_{i,\bar\ell_i}\big),
\label{eq:job-inertia}\\
E_{\mathcal U}(\bm q_{\mathcal U}\mid\bar{\bm\ell},\bm\theta)={}&
E^{\mathrm c}_{\mathcal U}(\bm q_{\mathcal U}\mid\bar{\bm\ell})
\nonumber\\
&+E^{\mathrm i}_{\mathcal U}(\bm q_{\mathcal U}\mid\bar{\bm\ell})
+\Xi_{\mathcal U}(\bm q_{\mathcal U})\nonumber\\
&
+\sum_{i\in\mathcal U}\sum_{k\in\Omega}\delta_{i,k}(\bm\theta)\,\tilde u_{i,k},
\label{eq:job-qubo}
\end{align}
\end{subequations}
returned in standard form $E_{\mathcal U}(\bm q_{\mathcal U})=\kappa_{\mathcal
U}+\sum_{\lambda}\alpha_{\lambda}q_{\lambda}+\sum_{\lambda<\lambda'}\omega_{\lambda\lambda'}q_{\lambda}q_{\lambda'}$ over its own bits $\bm
q_{\mathcal U}\in\{0,1\}^{|\mathcal U|(K-1)}$, where $\lambda$ and $\lambda'$ index
those bits. For later sums, extend the
quadratic coefficients symmetrically by $\omega_{\lambda'\lambda}:=\omega_{\lambda\lambda'}$ for $\lambda<\lambda'$.
Every job QUBO is built
solely from $\bar{\bm\ell}$, so its coefficients are independent of the other
current-round jobs. This independence allows the regional circuits to be
compiled and trained in parallel.

\subsection{Compiling and Training the Regional Circuit}
\label{subsec:circuit}

Each regional QUBO is compiled into a depth-$p$ QAOA circuit and trained
independently. Under the mapping $q_{\lambda}=(1-Z_{\lambda})/2$, the regional objective
$E_{\mathcal U}(\bm q_{\mathcal U})$ defines the diagonal cost operator
$\widehat E_{\mathcal U}$. With
$R_Z(\varphi)=e^{-\mathrm i\varphi Z/2}$ and
$R_{ZZ}(\varphi)=e^{-\mathrm i\varphi Z\otimes Z/2}$, the quadratic objective
in \eqref{eq:job-qubo} factors exactly into one $R_{ZZ}$ rotation per quadratic
term and one $R_{Z}$ rotation per qubit, absorbing the linear coefficient
of that qubit together with half of every quadratic coefficient touching
it, without requiring a cost table for circuit compilation:
\begin{equation}
\begin{aligned}
e^{-\mathrm i\gamma\widehat E_{\mathcal U}}={}&
\prod_{\lambda<\lambda'}R_{ZZ,\lambda\lambda'}\!
\Big(\tfrac{\gamma\omega_{\lambda\lambda'}}{2}\Big)\\
&\times\prod_{\lambda}R_{Z,\lambda}\!\Big({-\gamma}\Big(\alpha_{\lambda}
+\!\sum_{\lambda'\neq\lambda}\tfrac{\omega_{\lambda\lambda'}}{2}\Big)\Big).
\end{aligned}
\label{eq:phase-op}
\end{equation}
The regional QAOA state at depth $p$, with mixer $R_X(2\beta_r)$ applied
on every qubit, is
\begin{align}
|\bm\gamma,\bm\beta\rangle_{\mathcal U}&=
\prod_{r=1}^{p}e^{-\mathrm i\beta_r\widehat M_{\mathcal U}}
e^{-\mathrm i\gamma_r\widehat E_{\mathcal U}}\,
|+\rangle^{\otimes|\mathcal U|(K-1)},
\label{eq:qaoa-state}\\
\widehat M_{\mathcal U}&=\sum_{\lambda} X_{\lambda}.\nonumber
\end{align}

The $2p$ initial angles are sampled independently, with
$\gamma_r\sim\operatorname{Unif}(0,2\pi)$ and
$\beta_r\sim\operatorname{Unif}(0,\pi)$. COBYLA then minimizes the expected
value of $E_{\mathcal U}$. The register width is $n_{\mathcal U}=|\mathcal U|(K-1)$, bounded by
$Q_{\max}$ through \eqref{eq:zone-bound}. When $n_{\mathcal U}\le20$, training
uses the exact, zero-shot statevector expectation over an enumerated cost
table. Beyond that cutoff, exact simulation is intractable, so training
instead uses a finite-shot matrix-product-state estimate with
$\min(1024,\max(128,2^{\min(12,n_{\mathcal U})}))$ shots per COBYLA
evaluation. In either regime, training runs entirely on a local simulator,
after which the trained circuit is sent to its assigned backend for sampling.

\subsection{Regional Candidate Pool Construction}
\label{subsec:sampling}

Each backend converts a regional solution into an ordered candidate pool of
target size $N_{\mathrm c}\ge1$. The buses of every job $\mathcal U$ follow bus-index order.
Let $\mathcal S_{\mathcal U}$ be the multiset of measured registers
$\bm q_{\mathcal U}=(\bm q_i)_{i\in\mathcal U}$. A register is valid when every
$\bm q_i$ satisfies the validity condition of \eqref{eq:decode}, in which case
\eqref{eq:decode} leaves exactly one label active per bus and the job decoder
returns that label for every bus of the job,
\begin{equation}
\operatorname{dec}_{\mathcal U}(\bm q_{\mathcal U})=
\big(k_i\big)_{i\in\mathcal U}\in\Omega^{|\mathcal U|},
\qquad \tilde u_{i,k_i}=1 .
\label{eq:job-decoder}
\end{equation}
Invalid registers are discarded. For each assignment
$\bm a\in\Omega^{|\mathcal U|}$, define
\begin{subequations}
\label{eq:measured-pool}
\begin{align}
f_{\mathcal U}(\bm a)
&=\!\!\sum_{\bm q_{\mathcal U}\in\mathcal S_{\mathcal U}}\!\!
\ind\{\operatorname{dec}_{\mathcal U}(\bm q_{\mathcal U})=\bm a\},
\label{eq:candidate-frequency}\\
\mathcal C^{\mathrm m}_{\mathcal U}
&=\operatorname{top}^{\downarrow f_{\mathcal U}}_{N_{\mathrm c}}
\big\{\bm a\in\Omega^{|\mathcal U|}:f_{\mathcal U}(\bm a)>0\big\}.
\label{eq:measured-candidates}
\end{align}
\end{subequations}
Thus, distinct valid assignments are retained in decreasing measurement
frequency.

When fewer than $N_{\mathrm c}$ valid assignments are measured, distinct uniform
assignments complete the candidate pool. Let $A\ge N_{\mathrm c}$ denote
the maximum number of backfill attempts:
\begin{subequations}
\label{eq:candidate-backfill}
\begin{align}
\bm a_r&\overset{\mathrm{i.i.d.}}{\sim}
\operatorname{Unif}\big(\Omega^{|\mathcal U|}\big),\qquad
r=1,\ldots,A,
\label{eq:random-backfill}\\
\mathcal R_{\mathcal U}&=
\operatorname{uniq}(\bm a_1,\ldots,\bm a_A)
\setminus\mathcal C^{\mathrm m}_{\mathcal U},
\label{eq:backfill-set}\\
\widetilde{\mathcal C}_{\mathcal U}&=
\operatorname{first}_{N_{\mathrm c}}\big(
\mathcal C^{\mathrm m}_{\mathcal U}\Vert
\mathcal R_{\mathcal U}\big).
\label{eq:completed-pool}
\end{align}
\end{subequations}
Here, $\operatorname{uniq}$ removes duplicates while preserving first
occurrence, and $\Vert$ preserves left-to-right priority. The incumbent
regional assignment is then prepended as
\begin{equation}
\mathcal D_{\mathcal U}=
\operatorname{prepend}\!\left(
\bar{\bm\ell}|_{\mathcal U},\,
\widetilde{\mathcal C}_{\mathcal U}
\setminus\{\bar{\bm\ell}|_{\mathcal U}\}\right).
\label{eq:regional-candidate-pool}
\end{equation}
Here, $\operatorname{prepend}(\bm a,\mathcal C)$ places $\bm a$ at the head of
the ordered pool $\mathcal C$. The resulting pools are next evaluated on the
complete network.

\subsection{Candidate Evaluation and Ordering}
\label{subsec:scoring}

Complete candidates are ranked by feasibility and disruption before global
aggregation. The ordering preserves the feasible problem while still
providing direction among infeasible candidates. For an arbitrary
$\bm\ell\in\Omega^n$, fix one representative $\hat\imath_j\in\Gamma_j$ for every
coherent group. For $k\in\Omega$, $h\in\mathcal H$, and
$j\in\{1,\ldots,K\}$, define
\begin{subequations}
\label{eq:violation-indicators}
\begin{align}
N^h_k(\bm\ell)&=|\mathcal B_h\cap\mathcal B_k(\bm\ell)|
=\sum_{i\in\mathcal B_h}u_{i,k},
\label{eq:resource-count}\\
v^h_k(\bm\ell)&=\ind\{N^h_k(\bm\ell)<\mu_h\},\qquad
v^{\mathrm c}_k(\bm\ell)=\ind\{c_k(\bm\ell)\ne1\},
\label{eq:island-violation-indicators}\\
v^{\mathrm{coh}}_j(\bm\ell)&=
\ind\big\{|\{\ell_i:i\in\Gamma_j\}|>1\big\},
\nonumber\\
v^{\mathrm{sep}}(\bm\ell)&=
\ind\big\{|\{\ell_{\hat\imath_j}:j=1,\ldots,K\}|<K\big\}.
\label{eq:coherency-violation-indicators}
\end{align}
\end{subequations}

The total violated-clause count and the resulting big-constant score are
\begin{equation}
\nu(\bm\ell)=\sum_{k\in\Omega}
\left(\sum_{h\in\mathcal H}v^h_k(\bm\ell)+v^{\mathrm c}_k(\bm\ell)\right)
+\sum_{j=1}^{K}v^{\mathrm{coh}}_j(\bm\ell)+v^{\mathrm{sep}}(\bm\ell),
\label{eq:violation-count}
\end{equation}
\begin{equation}
\Psi(\bm\ell)=\Phi(\bm\ell)+\Lambda\,\nu(\bm\ell),
\qquad \Lambda=10^{6}.
\label{eq:score}
\end{equation}
$\Psi$ serves as a compact summary and as the score inside a single
exhaustively enumerated pool. A single large constant, however, cannot enforce
a strict priority among feasibility, cut cost, and violation count, so the
coordinator instead compares global candidates by the lexicographic tuple
\begin{equation}
\tau(\bm\ell)=
\Big(\ind\{\bm\ell\notin\mathcal X\},\
\begin{cases}\Phi(\bm\ell)&\bm\ell\in\mathcal X\\ \nu(\bm\ell)&\bm\ell\notin\mathcal X\end{cases},\
\Psi(\bm\ell),\ \mathrm{tie}(\bm\ell)\Big)
\label{eq:order}
\end{equation}
compared entrywise with smaller preferred, where $\mathrm{tie}(\bm\ell)$ is the
bus-index-ordered label vector that resolves exact ties deterministically. The
first entry ranks every feasible candidate ahead of every infeasible one,
after which \eqref{eq:order} reduces to $\Phi$ among feasible candidates and
to the violation count, then $\Psi$, among infeasible ones.

The anchor convention of Subsection~\ref{subsec:decomp} keeps
$v^{\mathrm{coh}}_j(\bm\ell)$ and $v^{\mathrm{sep}}(\bm\ell)$ identically zero
throughout the search. Both terms are retained in $\nu(\bm\ell)$ as general
well-formedness checks on an arbitrary assignment. The resulting order governs
both candidate aggregation and repair.

\subsection{Global Candidate Aggregation and Monotonic Update}
\label{subsec:aggregate}

The regional candidate pools are combined by beam aggregation, followed by
repair and monotonic commit. A fixed job ordering
$\mathcal U_1,\ldots,\mathcal U_{|\mathcal J|}$ makes aggregation independent
of worker completion times. For $\bm a\in\mathcal D_{\mathcal U}$,
$\operatorname{Embed}_{\mathcal U}(\bm\ell',\bm a)$ replaces the labels of
$\mathcal U$ in $\bm\ell'$ with $\bm a$ and leaves all other labels unchanged.
Each regional candidate is therefore incorporated as a complete block.
Beginning at the frozen snapshot, a bounded beam search folds every job pool
into the running frontier,
\begin{align}
\mathcal F_0&=\{\bar{\bm\ell}\},
\label{eq:beam}\\
\mathcal F_\sigma&=\operatorname{top}_{W}
\big\{\operatorname{Embed}_{\mathcal U_\sigma}(\bm\ell',\bm a):
\bm\ell'\in\mathcal F_{\sigma-1},\,\bm a\in\mathcal D_{\mathcal U_\sigma}\big\},
\nonumber
\end{align}
for $\sigma=1,\ldots,|\mathcal J|$, where $\operatorname{top}_{W}$
keeps the $W$ smallest-$\tau(\cdot)$ elements, ranked on
the complete network after every embedding, never on the local
value of any one job, since cross-job couplings and global constraints
make a sum of local values an invalid stand-in for global quality. Setting
$W\ge
\prod_{\mathcal U\in\mathcal J}|\mathcal D_{\mathcal U}|$ makes
\eqref{eq:beam} retain every assignment obtained by embedding one candidate
from each pool in the fixed job order, where a later job overwrites the labels
it shares with an earlier one. For a prescribed beam width, the number of candidate evaluations per
round is bounded by
$O(|\mathcal J|\,W
\max_{\mathcal U\in\mathcal J}|\mathcal D_{\mathcal U}|)$.

Postprocessing applies a constrained repair operator
$\operatorname{Rep}:\Omega^n\to\Omega^n$ to each aggregated candidate through greedy
single-bus descent and connectivity restoration~\cite{jiang2026regrid}. All
moves are restricted to $\mathcal B_{\mathrm v}$. The original candidate
remains the fallback under the selection
\begin{equation}
\operatorname{Post}(\bm\ell)=
\operatorname*{arg\,min}_{\bm y\in
\{\bm\ell,\operatorname{Rep}(\bm\ell)\}}\tau(\bm y).
\label{eq:postprocessing-selection}
\end{equation}
The best postprocessed candidate is accepted only when it improves the current
assignment,
\begin{subequations}
\label{eq:selection-commit}
\begin{align}
\bm\ell_{\mathrm c}^{(t)}&=\operatorname*{arg\,min}_{\bm y\in
\{\operatorname{Post}(\bm\ell'):\bm\ell'\in\mathcal F_{|\mathcal J|}\}}
\tau(\bm y),
\label{eq:selected-candidate}\\
\bm\ell^{(t+1)}&=
\begin{cases}
\bm\ell_{\mathrm c}^{(t)}, &
\tau(\bm\ell_{\mathrm c}^{(t)})<\tau(\bm\ell^{(t)}),\\
\bm\ell^{(t)}, & \text{otherwise},
\end{cases}
\label{eq:commit}
\end{align}
\end{subequations}
so a badly sampled round can never regress the committed assignment. A
separate archive $\bm\ell^{\mathrm{best}}$ retains the feasible assignment with the
smallest $\Phi$ seen so far, which need not solve \eqref{eq:program}. It is updated only when $\bm\ell^{(t+1)}$ is feasible and
improves $\Phi$. The final solution is $\bm\ell^{\mathrm{best}}$ when the archive is
nonempty and the best encountered assignment otherwise. Thus, an infeasible
or inferior round cannot overwrite the best feasible solution. The committed
assignment also determines the dual-price update used by the next round.

\subsection{Post-Commit Dual Coordination}
\label{subsec:duals}

The committed assignment supplies the global resource counts used to update
the next round's dual prices. By \eqref{eq:resource-count},
\eqref{eq:presence} is equivalent to $N^{h}_{k}(\bm\ell)\ge\mu_h$ for every
$k\in\Omega$ and $h\in\mathcal H$. Encoding this
inequality directly in a QUBO requires the binary slack register
\begin{align}
N^{h}_{k}(\bm\ell)-\!\!\sum_{r=0}^{S_h-1}\!\!2^{r}z^{h,r}_{k}=\mu_h,
\label{eq:slack}\\
S_h=\big\lceil\log_2(|\mathcal B_h|-\mu_h+1)\big\rceil,
\qquad z^{h,r}_k\in\{0,1\}.
\nonumber
\end{align}
Squaring this residual adds $K\sum_{h\in\mathcal H}S_h$ qubits and introduces couplings
between buses assigned to different regional jobs. These global constraints
are instead represented by the Lagrangian~\cite{jiang2026pace}
\begin{equation}
\Theta(\bm\ell,\bm\theta)=\Phi(\bm\ell)
+\sum_{k\in\Omega}\sum_{h\in\mathcal H}\theta_k^{h}
\big[\mu_h-N^{h}_{k}(\bm\ell)\big],
\label{eq:lagrangian}
\end{equation}
whose membership-dependent terms yield the regional coefficients in
\eqref{eq:dual-reward}. After commit, projected ascent updates each price as
\begin{align}
\theta_k^{h,(t+1)}=\Pi_{[0,\theta_{\max}]}
\Big(\theta_k^{h,(t)}+\eta
\big[\mu_h-N^{h}_{k}(\bm\ell^{(t+1)})\big]\Big),
\label{eq:dual-update}
\\
\Pi_{[0,\theta_{\max}]}(a)=
\min\{\theta_{\max},\max\{0,a\}\},
\nonumber
\end{align}
for $k\in\Omega$ and $h\in\mathcal H$, where $\eta>0$ and
$\theta_{\max}>0$. A deficit increases the corresponding
price, making that resource more attractive in the next round through
\eqref{eq:dual-reward}, whereas a satisfied requirement relaxes the price
toward zero. Updating only after \eqref{eq:commit} ensures that every job in a
round receives the same price vector. The next multipliers are computed from
the committed assignment, whether it is a newly accepted candidate or the
unchanged incumbent. This update closes the synchronous round.

The complete round-synchronous procedure is summarized in
Algorithm~\ref{alg:round}, from parallel regional solves to coordinator-side
aggregation and updates. Here, $T_{\min}$ and $T_{\max}$ are the minimum and
maximum numbers of rounds, $L$ is the stability threshold, $s$ counts
consecutive unchanged rounds, $\chi$ is the snapshot hash that identifies
stale worker results, and $\bot$ marks an empty archive.

\begin{algorithm}[t]
\caption{Round-Synchronous Distributed QAOA}
\label{alg:round}
\footnotesize
\begin{algorithmic}[1]
\Function{Round}{$\bm\ell,\bm\theta,t$}
    \State $\bar{\bm\ell}\gets\bm\ell$;\quad
    $\chi\gets\operatorname{hash}(\bar{\bm\ell})$
    \ForAll{$\mathcal U\in\mathcal J$ \textbf{concurrently}}
        \State assemble $E_{\mathcal U}(\cdot\mid\bar{\bm\ell},\bm\theta)$,
        \eqref{eq:job-qubo}
        \State train $(\bm\gamma,\bm\beta)$ by COBYLA (Sec.~\ref{subsec:circuit})
        \State $\mathcal D_{\mathcal U}\gets$ sample and backfill (Sec.~\ref{subsec:sampling})
    \EndFor
    \State discard any result tagged with round $\neq t$ or hash $\neq\chi$
    \State $\mathcal F_0\gets\{\bar{\bm\ell}\}$
    \For{$\sigma=1,\ldots,|\mathcal J|$}
        \State $\mathcal F_\sigma\gets$ \eqref{eq:beam} \Comment{fold in the pool of $\mathcal U_\sigma$}
    \EndFor
    \State \Return $\arg\min_{\bm y\in
    \{\operatorname{Post}(\bm\ell'):\bm\ell'\in\mathcal F_{|\mathcal J|}\}}
    \tau(\bm y)$
\EndFunction
\Statex
\State $\mathcal J\gets$ \eqref{eq:job-set} (Sec.~\ref{subsec:decomp})
\State $\bm\ell\gets\bm\ell^{(0)}$, \eqref{eq:init-label};\quad
$\bm\ell^{\mathrm{best}}\gets\bm\ell$ if $\bm\ell\in\mathcal X$ else $\bot$
\State $\bm\theta\gets\bm\theta^{(0)}$;\quad $s\gets0$
\For{$t=0,\ldots,T_{\max}-1$}
    \State $\bm\ell_{\mathrm c}\gets\Call{Round}{\bm\ell,\bm\theta,t}$
    \If{$\tau(\bm\ell_{\mathrm c})<\tau(\bm\ell)$}
        \State $\bm\ell\gets\bm\ell_{\mathrm c}$;\quad $s\gets0$
    \Else
        \State $s\gets s+1$
    \EndIf
    \If{$\bm\ell\in\mathcal X$ \textbf{and} ($\bm\ell^{\mathrm{best}}=\bot$ \textbf{or} $\Phi(\bm\ell)<\Phi(\bm\ell^{\mathrm{best}})$)}
        \State $\bm\ell^{\mathrm{best}}\gets\bm\ell$
    \EndIf
    \State advance $\bm\theta$ by \eqref{eq:dual-update} \Comment{if multipliers are priced}
    \If{$s\ge L$ \textbf{and} $t+1\ge T_{\min}$}\ \textbf{break}\EndIf
\EndFor
\State \Return $\bm\ell^{\mathrm{best}}$ if $\bm\ell^{\mathrm{best}}\neq\bot$ else $\bm\ell$
\end{algorithmic}
\end{algorithm}

\subsection{Monotonicity and Completeness Analysis}
\label{subsec:guarantees}

The procedure's monotonic commit rule and exhaustive-aggregation limit yield
two formal properties.

\begin{proposition}[Rank monotonicity]
\label{prop:monotone}
Along the trajectory $\{\bm\ell^{(t)}\}$ produced by Algorithm~\ref{alg:round},
$\tau(\bm\ell^{(t)})$ is non-increasing in $t$. Once $\bm\ell^{\mathrm{best}}$ first
attains $\Phi(\bm\ell^{\mathrm{best}})=\Phi^{\star}$, it is never subsequently
replaced.
\end{proposition}
\begin{IEEEproof}
By \eqref{eq:commit}, $\bm\ell^{(t+1)}$ equals $\bm\ell^{(t)}$ unless
$\tau(\bm\ell_{\mathrm c}^{(t)})<\tau(\bm\ell^{(t)})$, in which case
$\bm\ell^{(t+1)}=\bm\ell_{\mathrm c}^{(t)}$ attains the strictly smaller value,
so $\tau(\bm\ell^{(t)})$ is non-increasing. The archive is overwritten only
by a feasible assignment with strictly smaller $\Phi$, and every feasible
assignment satisfies $\Phi(\bm\ell)\ge\Phi^{\star}$, so once
$\Phi(\bm\ell^{\mathrm{best}})=\Phi^{\star}$ no candidate can replace it.
\end{IEEEproof}

\begin{proposition}[Completeness under exhaustive coverage]
\label{prop:complete}
If $W\ge\prod_{\mathcal U\in\mathcal J}|\mathcal D_{\mathcal U}|$,
then $\mathcal F_{|\mathcal J|}$ from \eqref{eq:beam} contains every complete
assignment obtained by embedding one candidate from each job pool in the fixed
job order. If, in
addition, every job pool is exhaustive,
$\mathcal D_{\mathcal U}=\Omega^{|\mathcal U|}$ for all
$\mathcal U\in\mathcal J$, then the best feasible element of
$\mathcal F_{|\mathcal J|}$ solves \eqref{eq:program} exactly.
\end{proposition}
\begin{IEEEproof}
With $W\ge
\prod_{\mathcal U\in\mathcal J}|\mathcal D_{\mathcal U}|$, no
candidate is ever discarded by $\operatorname{top}_{W}$ at any
fold step of \eqref{eq:beam}. Thus, the final frontier contains every complete
assignment reachable by that ordered embedding. Because the zones cover
$\mathcal B_{\mathrm v}$ and every job pool is exhaustive, selecting from each
job the restriction of a target assignment to that job reproduces the target
after the final embedding, so the frontier contains an anchored representative
of every physically distinct partition.
The objective $\Phi$ is invariant under label permutations, so its best
feasible representative attains $\Phi^{\star}$.
\end{IEEEproof}

Proposition~\ref{prop:complete} applies exactly to the smallest benchmark
instance. Larger instances use a finite $W$, which bounds the per-round cost
while Proposition~\ref{prop:monotone} still holds, so the committed sequence
improves monotonically and the archive keeps the best feasible assignment.
Section~\ref{sec:results} measures the resulting quality against the Gurobi
optimum.

\section{Numerical Examples}
\label{sec:results}

In this section, numerical examples are carried out on power-system test systems spanning networks from 9 to 300 buses. The proposed framework is benchmarked across a range of quantum-computing environments, including the IBM real quantum computer \texttt{ibm\_marrakesh}, the IBM ideal simulator, the IBM ideal simulator equipped with a calibrated noise model derived from the \texttt{ibm\_marrakesh} backend, the 108-qubit Rigetti Cepheus-1-108Q QPU, the 54-qubit IQM Emerald QPU, and the 34-qubit Amazon Braket SV1 state-vector simulator. All environments cover the full 9--300-bus range except Rigetti, which hardware limits confine to systems through 118 buses, so the comparison spans increasing scale under both idealized and hardware-realistic execution.

Table~\ref{tab:experiment-parameters} summarizes the experimental setup. Here,
$N_{\mathrm w}$ denotes the number of parallel workers,
$|\mathcal S_{\mathcal U}|$ is the number of measurement shots per regional
job, and $T_{\min}/T_{\max}$ gives the minimum/maximum numbers of coordination
rounds. Across the 9--300-bus systems, the regional circuit-width limit
remains $Q_{\max}\le 20$. The QPU runs of the 118--300-bus systems used a
smaller coordination budget, with $T_{\max}=8$, $L\le2$, and
$N_{\mathrm w}=4$, and 3000 shots on the 300-bus system.

\begin{table}[!t]
\caption{Key parameters for the numerical experiments.}
\label{tab:experiment-parameters}
\centering
\scriptsize
\renewcommand{\arraystretch}{1.05}
\begin{tabular}{c c r c c c c}
\toprule
Test system & $p$ & $|\mathcal S_{\mathcal U}|$ & $N_{\mathrm w}$ & $Q_{\max}$ & $T_{\min}/T_{\max}$ & $L$ \\
\midrule
9-, 14-bus & 1 & 100  & 4  & 4  & 1/1  & 1 \\
24-bus  & 1 & 150  & 4  & 8  & 1/1  & 1 \\
30-bus  & 1 & 150  & 4  & 4  & 1/1  & 1 \\
39-bus  & 1 & 300  & 4  & 8  & 1/2  & 2 \\
57-bus  & 3 & 500  & 4  & 8  & 1/3  & 3 \\
73-bus  & 2 & 500  & 4  & 8  & 1/3  & 3 \\
89-bus  & 3 & 1000 & 4  & 16 & 1/3  & 3 \\
118-bus & 3 & 3000 & 16 & 20 & 1/20 & 4 \\
145-bus & 3 & 3000 & 24 & 20 & 1/20 & 4 \\
300-bus & 4 & 4000 & 24 & 20 & 1/20 & 4 \\
\bottomrule
\end{tabular}
\end{table}

\subsection{Solution Quality and Resource Efficiency}
\label{subsec:solution-resource-results}

This subsection demonstrates the effectiveness of the proposed method through
comparison with monolithic QAOA in the IBM quantum computer backend. Table~\ref{tab:qubit-savings} compares their
circuit width, two-qubit gates, circuit depth, and QPU time. Each paired entry
lists the proposed and monolithic values. Here, $Q$ is the circuit width in
qubits, namely $\max_{\mathcal U\in\mathcal J}n_{\mathcal U}$ for the proposed
method and $Q_{\mathrm G}$ for the monolithic circuit, while $G$, $d$, and $t_q$
denote the maximum transpiled two-qubit gate count, maximum circuit depth, and
accumulated QPU time. The qubit-width saving and reduction factor are
$\Delta_Q\triangleq100(1-\max_{\mathcal U}n_{\mathcal U}/Q_{\mathrm G})$ and
$F_Q\triangleq Q_{\mathrm G}/\max_{\mathcal U}n_{\mathcal U}$.
Across the 9--89-bus systems,
decomposition reduces width by
33.3--89.6\% and two-qubit gates by 59.3--96.4\%, while reducing depth in six
of the eight cases. For the 118--300-bus systems, the width saving increases
to 90.6--95.7\%, enabling quantum-hardware execution with at most 20 qubits
when the corresponding monolithic circuits are unavailable.

\begin{table}[!t]
\caption{Quantum-hardware resource comparison with monolithic QAOA.}
\label{tab:qubit-savings}
\centering
\scriptsize
\renewcommand{\arraystretch}{1.05}
\setlength{\tabcolsep}{2.5pt}
\resizebox{0.90\columnwidth}{!}{%
\begin{tabular}{c r r r r r r}
\toprule
& \multicolumn{4}{c}{Proposed/monolithic} & & \\
\cmidrule(lr){2-5}
Test system & $Q$ & $G$ & $d$ & $t_q$ (s) & $\Delta_Q$ (\%) & $F_Q$ \\
\midrule
9-bus   & 4/6    & 6/21      & 29/25     & 2/2      & 33.3 & $1.50\times$ \\
14-bus  & 4/9    & 11/27     & 42/77     & 2/2      & 55.6 & $2.25\times$ \\
24-bus  & 8/26   & 80/248    & 160/94    & 3/3      & 69.2 & $3.25\times$ \\
30-bus  & 4/24   & 11/127    & 47/62     & 3/4      & 83.3 & $6.00\times$ \\
39-bus  & 8/58   & 81/2227   & 162/720   & 8/6      & 86.2 & $7.25\times$ \\
57-bus  & 8/50   & 85/638    & 177/544   & 18/6     & 84.0 & $6.25\times$ \\
73-bus  & 8/54   & 124/2575  & 275/719   & 35/26    & 85.2 & $6.75\times$ \\
89-bus  & 16/154 & 315/3129  & 530/865   & 72/47    & 89.6 & $9.63\times$ \\
118-bus & 18/192 & 995/NA    & 1434/NA   & 120/NA   & 90.6 & $10.67\times$ \\
145-bus & 18/285 & 2613/NA   & 2850/NA   & 277/NA   & 93.7 & $15.83\times$ \\
300-bus & 20/462 & 1266/NA   & 1611/NA   & 520/NA   & 95.7 & $23.10\times$ \\
\bottomrule
\end{tabular}
}
\end{table}

The resource reductions preserve solution quality: every reported backend in
Table~\ref{tab:platform-comparison} attains the Gurobi optimum
$\Phi=\Phi^\star$. Fig.~\ref{fig:islanding-solutions} presents the four largest
solutions. In each case, every island is physically admissible and self-sustaining after the cuts are made. The 89-, 118-, 145-,
and 300-bus solutions cut only 46/206, 17/179, 89/422, and 23/409 branches,
respectively.
Thus, the proposed method combines substantial resource savings with
optimal and physically admissible islanding solutions over the full test
range.

\begin{figure}[!t]
\centering
\includegraphics[width=\columnwidth]{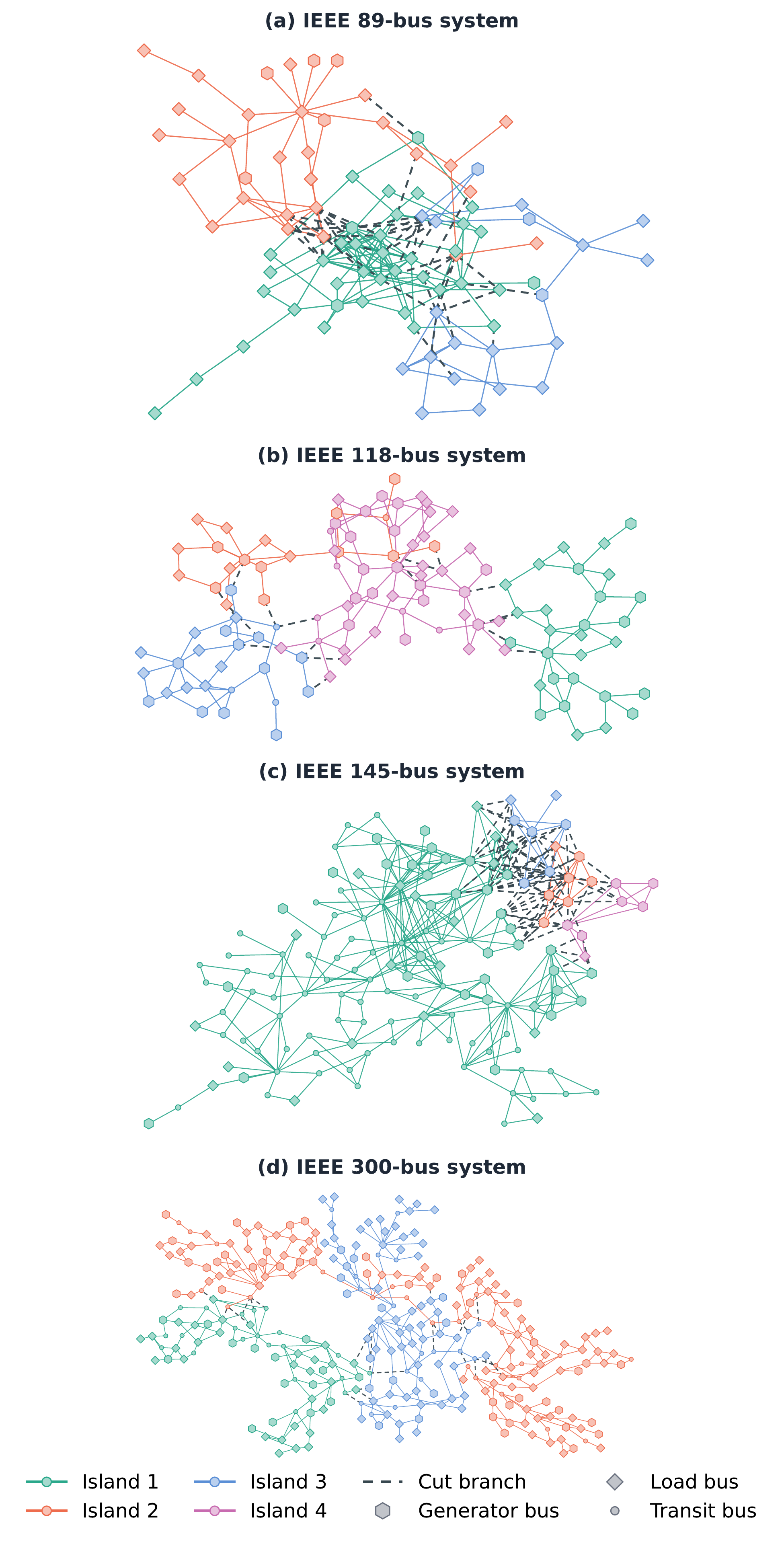}
\caption{Islanding solutions for the (a)~IEEE 89-bus, (b)~IEEE 118-bus,
(c)~IEEE 145-bus, and (d)~IEEE 300-bus systems.}
\label{fig:islanding-solutions}
\end{figure}

\subsection{Noise Resilience and Cross-Backend Benchmarking}
\label{subsec:noise-benchmarking}

In this subsection, the noise resilience of the proposed method is assessed by
executing the same decomposed workload on ideal simulators, noisy simulators,
and real quantum hardware. Unlike
Section~\ref{subsec:solution-resource-results}, where the proposed method is
measured against monolithic QAOA, the workload here is held fixed while the
noise regime and the native architecture vary, which makes the cross-backend
benchmark the instrument through which resilience becomes visible. Noise and
hardware execution act on the best feasible cut objective $\Phi$, and
transpilation onto each native gate set and coupling map determines the
maximum ISA circuit depth $d$ and the two-qubit gate count $G$.
Table~\ref{tab:platform-comparison} and Fig.~\ref{fig:backend-resources}
summarize both effects over all available system--backend pairs, with the
Gurobi optimum $\Phi^\star$ serving as the classical reference. Both
$\Phi$ and $\Phi^\star$ are reported in MW, consistent with the line
weights defined in \eqref{eq:line-weight}.

\begin{table}[!t]
\caption{Cross-platform cut objectives and circuit resources.}
\label{tab:platform-comparison}
\centering
\scriptsize
\renewcommand{\arraystretch}{0.80}
\setlength{\tabcolsep}{3.0pt}
\resizebox{0.82\columnwidth}{!}{%
\begin{tabular}{c l r r r r}
\toprule
Test system & Backend & $\Phi$ (MW) & $\Phi^\star$ (MW) & $d$ & $G$ \\
\midrule
9-bus  & Aer ideal   & 1.238 & 1.238 & 6  & 3 \\
       & Aer noise   & 1.238 &       & 30 & 6 \\
       & IBM QPU     & 1.238 &       & 29 & 6 \\
       & SV1         & 1.238 &       & 13 & 6 \\
       & IQM QPU     & 1.238 &       & 19 & 6 \\
       & Rigetti QPU & 1.238 &       & 22 & 6 \\
\addlinespace[0pt]
14-bus & Aer ideal   & 88.241 & 88.241 & 7  & 4 \\
       & Aer noise   & 88.241 &        & 45 & 11 \\
       & IBM QPU     & 88.241 &        & 42 & 11 \\
       & SV1         & 88.241 &        & 16 & 8 \\
       & IQM QPU     & 88.241 &        & 18 & 8 \\
       & Rigetti QPU & 88.241 &        & 19 & 8 \\
\addlinespace[0pt]
24-bus & Aer ideal   & 776.206 & 776.206 & 10  & 20 \\
       & Aer noise   & 776.206 &         & 128 & 78 \\
       & IBM QPU     & 776.206 &         & 160 & 80 \\
       & SV1         & 776.206 &         & 25  & 40 \\
       & IQM QPU     & 776.206 &         & 31  & 40 \\
       & Rigetti QPU & 776.206 &         & 34  & 40 \\
\addlinespace[0pt]
30-bus & Aer ideal   & 16.863 & 16.863 & 7  & 4 \\
       & Aer noise   & 16.863 &        & 46 & 11 \\
       & IBM QPU     & 16.863 &        & 47 & 11 \\
       & SV1         & 16.863 &        & 16 & 8 \\
       & IQM QPU     & 16.863 &        & 20 & 8 \\
       & Rigetti QPU & 16.863 &        & 22 & 8 \\
\addlinespace[0pt]
39-bus & Aer ideal   & 228.993 & 228.993 & 10  & 20 \\
       & Aer noise   & 228.993 &         & 169 & 83 \\
       & IBM QPU     & 228.993 &         & 162 & 81 \\
       & SV1         & 228.993 &         & 25  & 40 \\
       & IQM QPU     & 228.993 &         & 33  & 40 \\
       & Rigetti QPU & 228.993 &         & 37  & 40 \\
\addlinespace[0pt]
57-bus & Aer ideal   & 128.239 & 128.239 & 21  & 24 \\
       & Aer noise   & 128.239 &         & 173 & 85 \\
       & IBM QPU     & 128.239 &         & 177 & 85 \\
       & SV1         & 128.239 &         & 56  & 48 \\
       & IQM QPU     & 128.239 &         & 66  & 48 \\
       & Rigetti QPU & 128.239 &         & 85  & 54 \\
\addlinespace[0pt]
73-bus & Aer ideal   & 284.543 & 284.543 & 18  & 32 \\
       & Aer noise   & 284.543 &         & 284 & 123 \\
       & IBM QPU     & 284.543 &         & 275 & 124 \\
       & SV1         & 284.543 &         & 48  & 64 \\
       & IQM QPU     & 284.543 &         & 59  & 64 \\
       & Rigetti QPU & 284.543 &         & 64  & 64 \\
\addlinespace[0pt]
89-bus & Aer ideal   & 3790.938 & 3790.938 & 62   & 240 \\
       & Aer noise   & 3790.938 &          & 434 & 279 \\
       & IBM QPU     & 3790.938 &          & 530  & 315 \\
       & SV1         & 3790.938 &          & 179  & 480 \\
       & IQM QPU     & 3790.938 &          & 207  & 480 \\
       & Rigetti QPU & 3790.938 &          & 220  & 480 \\
\addlinespace[0pt]
118-bus & Aer ideal   & 634.050 & 634.050 & 58   & 189 \\
        & Aer noise   & 634.050 &         & 1374 & 1006 \\
        & IBM QPU     & 634.050 &         & 1434 & 995  \\
        & SV1         & 634.050 &         & 167  & 378 \\
        & IQM QPU     & 634.050 &         & 185  & 378 \\
        & Rigetti QPU & 634.050 &         & 193  & 378 \\
\addlinespace[0pt]
145-bus & Aer ideal & 11859.800 & 11859.800 & 98   & 459  \\
        & Aer noise & 11859.800 &           & 2769 & 2372 \\
        & IBM QPU   & 11859.800 &           & 2850 & 2613 \\
        & SV1       & 11859.800 &           & 287 & 918 \\
        & IQM QPU   & 11859.800 &           & 297 & 918 \\
\addlinespace[0pt]
300-bus & Aer ideal & 2322.400 & 2322.400 & 74   & 232  \\
        & Aer noise & 2322.400 &          & 1528 & 1097 \\
        & IBM QPU   & 2322.400 &          & 1611 & 1266 \\
        & SV1       & 2322.400 &          & 214 & 464 \\
        & IQM QPU   & 2322.400 &          & 243 & 464 \\
\bottomrule
\end{tabular}
}
\end{table}

\begin{figure}[!t]
\centering
\includegraphics[width=\columnwidth]{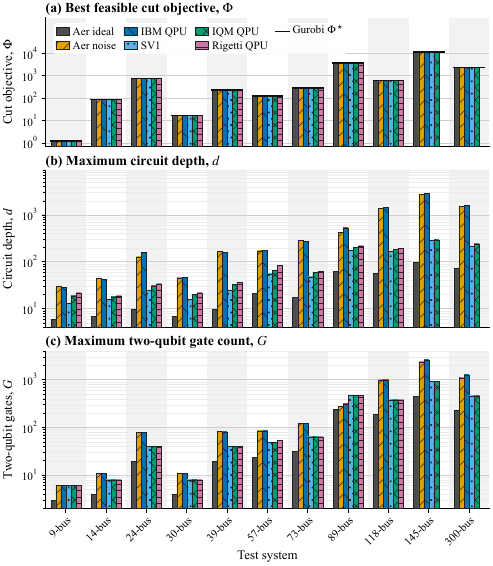}
\caption{Best feasible (a)~cut objective $\Phi$ and maximum (b)~ISA circuit
depth $d$ and (c)~two-qubit gate count $G$ across backends and the
9--300-bus systems.}
\label{fig:backend-resources}
\end{figure}

Across every execution environment considered, from ideal simulation through
the calibrated noise model to the IBM, IQM, and Rigetti QPUs, the recovered
solution quality is unchanged, each reported backend attaining
$\Phi=\Phi^\star$ at all system sizes available on that platform. Such
invariance indicates that the complete workflow preserves the best feasible
solution despite changes in sampling and device noise. This behavior is consistent with the
regional candidate pools and feasibility-first coordinator, which can absorb
perturbations in sample frequencies when an optimal candidate remains
available for aggregation. The equality of $\Phi$ across backends therefore
characterizes the recovered feasible solution rather than the sampled
distributions, which need not coincide.

The resource data show a substantially different form of backend
sensitivity, and the informative quantity is how the compilation overhead
relative to Aer ideal scales with system size. For SV1 and IQM that overhead
is constant in the gate count, which equals exactly twice the ideal value on
all eleven test systems, and Rigetti departs from it only on the 57-bus case
with 54 gates against 48, so each modelled interaction compiles into a fixed
pair of native two-qubit gates regardless of problem size. Depth behaves
almost as steadily, the SV1 factor rising only from 2.2 to 2.9 between the 9-
and 300-bus systems while IQM stays between 2.6 and 3.3 and Rigetti between
2.7 and 4.0. Compilation for the IBM devices scales with the problem instead,
its gate factor reaching 5.3 to 5.7 and its depth factor 21.8 to 29.1 on the
three largest systems, against 2.0 and 4.8 on the 9-bus system, and Aer noise
reproduces both factors throughout. The resulting gap
to the other backends therefore widens with network size, reaching 7.5 to 9.9
in depth and 2.6 to 2.8 in gate count over the 118--300-bus systems. Depth therefore varies far more across backends than gate count, so the
critical path is governed by routing and serialization rather than by the
number of two-qubit operations. Depth also matters more on noisy hardware,
since it sets how long the state is exposed to decoherence.

Repeated gate counts combined with different depths separate the
interaction workload from the way it is scheduled. SV1 executes without a
coupling constraint, so its depth is the shallowest realization of that
workload and serves as the reference for the two QPUs that carry the same
$G$. IQM stays between 3 and 16 percent above this reference for the
89--300-bus systems, which is the cost of routing and native-gate translation
on the Emerald lattice, whereas Rigetti pays more for identical work, needing
85 layers against 56 on the 57-bus system and 220 against 179 on the 89-bus
system. This penalty also narrows as the regions grow, falling from 46 percent on
the 9-bus system to 14 percent on the 300-bus system, which suggests that the
fixed layout cost is amortized once the regional circuits approach the width
cap.

Neither resource metric is monotone in network size. Most notably, the
145-bus case is more demanding than the 300-bus case on all five backends
reported for both systems. On IQM, $d$ decreases from 297 to 243 and $G$ from 918
to 464, and SV1 follows with 287 to 214 and the same halving of $G$. The IBM
QPU repeats the pattern on a larger scale, with $d$ falling from 2850 to 1611,
as do Aer ideal and Aer noise. Because Table~\ref{tab:platform-comparison} reports the
largest regional circuit rather than the aggregate workload, its resources
follow the width, coupling density, and transpiled topology of the hardest
regional QUBO, not the bus count. A larger network can thus hold more jobs yet
need a cheaper maximum circuit, and the next subsection analyzes the total
workload behind these maxima.

\subsection{Computational Complexity Analysis}
\label{subsec:complexity}

Because \eqref{eq:job-set} is a union, $|\mathcal J|\le M+R$, with equality
unless a retained window coincides with a zone. Since each cross-zone line
generates at most one retained window,
\begin{equation}
|\mathcal J|\le
\left\lceil\frac{n_{\mathrm v}}{\zeta_{\max}}\right\rceil+|\mathcal L|.
\label{eq:job-complexity}
\end{equation}
Because each bus uses $K-1$ qubits, $K$ affects $|\mathcal J|$ through
$\zeta_{\max}=\lfloor Q_{\max}/(K-1)\rfloor$. Fixed $K$ and $Q_{\max}$
therefore affect constants, but not the asymptotic order.
Multi-source initialization, linear-time zone formation, and sorting the
cross-zone lines require $O((n+|\mathcal L|)\log n)$ preprocessing time.

For each $\mathcal U\in\mathcal J$, the circuit width is
$n_{\mathcal U}=|\mathcal U|(K-1)\le Q_{\max}$. One QAOA layer contains at
most $n_{\mathcal U}(n_{\mathcal U}-1)/2$ two-qubit cost rotations and
$2n_{\mathcal U}$ one-qubit cost and mixer rotations. Including state
preparation,
\begin{align}
\operatorname{gates}(\mathcal U)
&\le n_{\mathcal U}
+p\left[\frac{n_{\mathcal U}(n_{\mathcal U}-1)}{2}
+2n_{\mathcal U}\right]\nonumber\\
&=O(pQ_{\max}^{2}),
\label{eq:gate-complexity}\\
\operatorname{depth}(\mathcal U)
&=O\!\left(pn_{\mathcal U}^{2}\right)
\le O\!\left(pQ_{\max}^{2}\right).
\label{eq:depth-complexity}
\end{align}
Equations~\eqref{eq:gate-complexity}--\eqref{eq:depth-complexity} bound one
regional job. As shown in Table~\ref{tab:platform-comparison}, the 300-bus
case has lower $d$ and $G$ than the 145-bus case on every common backend.
This comparison does not imply lower end-to-end cost because the total
workload accumulates all regional jobs, shots, and coordination rounds.
On quantum hardware, one circuit execution produces one bitstring in
$\mathcal S_{\mathcal U}$. The logical-gate workload per round is
\begin{equation}
O\!\left(p\sum_{\mathcal U\in\mathcal J}
|\mathcal S_{\mathcal U}|n_{\mathcal U}^{2}\right)
\le O\!\left(pQ_{\max}^{2}
\sum_{\mathcal U\in\mathcal J}|\mathcal S_{\mathcal U}|\right).
\label{eq:qpu-gate-executions}
\end{equation}
Decoding all measured bitstrings takes
$O(Q_{\max}\sum_{\mathcal U}|\mathcal S_{\mathcal U}|)$ operations per round,
where $\mathcal U\in\mathcal J$.

Since $|\mathcal D_{\mathcal U}|\le N_{\mathrm c}+1$, beam aggregation evaluates
$O(|\mathcal J|W(N_{\mathrm c}+1))$ candidates per round, each requiring
$O(n+|\mathcal L|)$ operations to compute $\tau$, for
$O(|\mathcal J|W(N_{\mathrm c}+1)(n+|\mathcal L|))$ operations in total.

For sparse networks with $|\mathcal L|=O(n)$ and fixed $K$, $p$, $Q_{\max}$,
$W$, $N_{\mathrm c}$, and $|\mathcal S_{\mathcal U}|$, $|\mathcal J|=O(n)$ and
$\sum_{\mathcal U\in\mathcal J}|\mathcal S_{\mathcal U}|=O(n)$. Therefore,
\begin{align}
\text{quantum workload per round}&=O(n),
\label{eq:sparse-quantum-complexity}\\
\text{coordination per round}&=O(n^{2}),
\label{eq:sparse-classical-complexity}\\
\text{overall complexity}
&=O\!\left(n\log n+T_{\max}n^{2}\right)\nonumber\\
&=O(T_{\max}n^{2}).
\label{eq:overall-complexity}
\end{align}
Thus, regional decomposition keeps the quantum workload linear and the overall
complexity polynomial.

\section{Conclusion}
\label{sec:conclusion}

In this study, a round-synchronous distributed QAOA framework is developed for
coherent controlled islanding under a prescribed per-circuit qubit budget.
Round synchronization confines the quantum effort to bounded regional
subproblems that are solved concurrently, while global feasibility and
monotone improvement are enforced classically over the complete network. Circuit width is thereby bounded independently of
network size, lowering the qubit requirement by 33.3 to 95.7 percent and the
two-qubit gate count by up to 96.4 percent relative to monolithic QAOA at an
overall complexity of $O(T_{\max}n^{2})$ for sparse networks. Benchmarking on
six execution environments, from ideal simulation to calibrated noise and
different quantum processors, reaches the Gurobi optimum throughout the
9--300-bus range. The proposed method therefore combines noise resilience with
strong adaptability across quantum backends, retaining optimal solution quality even though compilation
cost varies by nearly an order of magnitude between them. Since only the bounded
regional circuits require compilation, the framework transfers to new
platforms without redesign, and the comparison reported here is among the
first cross-platform benchmarks of gate-based quantum optimization for power
systems.

In future work, error mitigation, adaptive shot allocation, and a larger
per-circuit width as hardware matures will be explored to reduce the
compilation overhead measured here and to extend the reachable network size.

\balance
\bibliographystyle{IEEEtran}
\bibliography{References}

\end{document}